\documentclass[11pt,twoside,leqno]{article}  

\usepackage[hmargin=1.1in,vmargin=1.1in]{geometry}

\usepackage{graphicx}
\usepackage{amsthm}
\usepackage{amssymb}
\usepackage{amsmath}
\usepackage{setspace}
\usepackage{cite}
\usepackage[usenames,dvipsnames]{color}

\def\final{0}  
\def\iflong{\iffalse}
\ifnum\final=0  

\newcommand{\nnote}[1]{[{\tiny navin: \bf #1}]\marginpar{*}}

\else 
\fi  

\newtheorem{theorem}{Theorem}
\newtheorem{lemma}[theorem]{Lemma}
\newtheorem{corollary}[theorem]{Corollary}
\newtheorem{prop}{Proposition}
\newtheorem{definition}{Definition}

\renewcommand{\qed}{\nobreak \ifvmode \relax \else
      \ifdim\lastskip<1.5em \hskip-\lastskip
      \hskip1.5em plus0em minus0.5em \fi \nobreak
      \vrule height0.75em width0.5em depth0.25em\fi}

\def\eps{\epsilon}
\def\vbl{\mathrm{vbl}}

\def\P{\mathcal{P}}
\def\A{\mathcal{A}}

\newcommand{\prob}[1]{{\sf Pr}\left(#1\right)}
\newcommand{\ceil}[1]{\lceil #1 \rceil}

\newcommand{\asabove}{}

\begin{document}

\title{\Large Deterministic Algorithms for the Lov\'asz Local Lemma\footnote{A preliminary version of this work appeared in the ACM-SIAM Symposium on Discrete Algorithms, 2010.}}
\author{Karthekeyan Chandrasekaran \thanks{Georgia Institute of Technology; This work was done while visiting Microsoft Research, India.} \\ \small karthe@gatech.edu \and Navin Goyal \thanks{Microsoft Research, India.} \\ \small navingo@microsoft.com \and Bernhard Haeupler \thanks{Massachusetts Insititute of Technology, Computer Science and Artificial Intelligence Lab, 32 Vassar Street, Cambridge MA 02139; This work was partially supported by an MIT Presidential Fellowship from Akamai. This work was done while visiting Microsoft Research, India.}\\ \small haeupler@mit.edu }
\date{}

\maketitle


\begin{abstract}
The Lov\'asz Local Lemma~\cite{ErdosLovasz} (LLL) is a powerful result in probability theory that informally states the following: the probability that none of a set of bad events happens is positive if the probability of each event is small 
compared to the number of events that depend on it. The LLL is often
used for non-constructive existence proofs of combinatorial
structures. A prominent application is to $k$-CNF formulas, where the
LLL implies that if every clause in a formula shares variables with at
most $d \leq 2^k/e-1$ other clauses then such a formula has a
satisfying assignment.  Recently, a randomized algorithm to
efficiently construct a satisfying assignment in this setting was
given by Moser \cite{Moser08B}. Subsequently Moser and Tardos
\cite{MoserTardos} gave a general algorithmic framework for the LLL
and a randomized algorithm within this framework to construct the structures guaranteed by
the LLL. The main problem left open by Moser
and Tardos was to design an efficient deterministic 
algorithm for constructing structures guaranteed by the LLL. In this paper we
provide such an algorithm. Our algorithm
works in the general framework of Moser--Tardos with a minimal loss in
parameters. For the special case of constructing satisfying
assignments for $k$-CNF formulas with $m$ clauses, where each clause
shares variables with at most $d \leq 2^{k/(1+\eps)}/e - 1$ other clauses,
for any $\eps\in (0,1)$, we give a deterministic algorithm that finds
a satisfying assignment in time $\tilde{O}(m^{2(1+1/\eps)})$. This
improves upon the deterministic algorithms of Moser and of
Moser--Tardos with running times $m^{\Omega(k^2)}$ and $m^{\Omega(d
\log d)}$ respectively, which are super-polynomial for $k=\omega(1)$ and $d=\omega(1)$, and upon the previous best deterministic algorithm of Beck which runs in polynomial time only for $d\leq 2^{k/16}/4$. Our algorithm is the first deterministic algorithm that works in the general framework of Moser--Tardos.  We also give a parallel NC algorithm for the same setting improving upon an algorithm of Alon \cite{Alon}.
\end{abstract}

\section{Introduction}

The Lov\'asz Local Lemma~\cite{ErdosLovasz} (henceforth LLL) informally states that the probability that none of a set of bad events happens is nonzero if the probability of each event is small compared to the number of events that depend on it (see Section \ref{sec:nonconstructive-LLL} for details). It is a powerful result in probability theory and is often used in conjunction with the probabilistic method to prove the existence of combinatorial structures. For this, one designs a random process guaranteed to generate the desired structure if none of a set of bad events happens. If the events satisfy the above-mentioned condition, then the LLL guarantees that the probability that the random process builds the desired structure is positive, thereby implying its existence. 
For many applications of the LLL, it is also important to find the desired structures efficiently.
Unfortunately, the original proof of the LLL \cite{ErdosLovasz} does not lead to an efficient algorithm.  
In many applications of the LLL, the probability of none of the bad events happening is negligible. Consequently, the same random process does not directly provide a randomized algorithm to find the desired structure. 
Further, 
in most applications where the LLL is useful (e.g., \cite{Feige, LeightonMaggsRao, MolloyReedBook}), the proof of existence of the desired structure is known only through the LLL (one exception to this is
\cite{Feige}). 
Thus, an efficient algorithm for the LLL would also lead to an efficient algorithm to find these desired structures.
Starting with the work of Beck~\cite{Beck}, a number of papers, e.g. \cite{Alon, CzumajS, MolloyReed, Srinivasan, Moser08A}, have sought to make the LLL algorithmic. Before discussing 
these results in more detail we describe the LLL formally. 

\subsection{The Lov\'asz Local Lemma}\label{sec:nonconstructive-LLL}
The Lov\'asz Local Lemma gives a lower bound on the probability of
avoiding a possibly large number of ``bad'' events that are not ``too
dependent'' on each other. Let $\A$ be a finite set of events in a
probability space. Let $G$ be an undirected graph on vertex set $\A$
with the property that every event $A \in \A$ is mutually
independent\footnote{An event $A$ is mutually independent of a set of
events $\{B_1, B_2, \ldots\}$ if $\prob{A} = \prob{A \ | \
f(B_1,B_2, \ldots)}$ for every function $f$ that can be expressed using finitely many unions and intersections of the arguments.} of the set of all events not in its neighborhood. We assume throughout that $G$ does not contain any self-loops. We denote the set of neighbors of an event $A$ by $\Gamma(A)$, i.e., $\Gamma(A) := \{B \in \A \ \ | \ \ \{A,B\}
 \in E(G) \}$. 
The general version of the LLL is the following.

\begin{theorem}\cite{ErdosLovasz,spencer1975ramsey} \label{thm:nonconstructive-lll-general}
For $\A$ and $G$ as defined above, suppose there exists an assignment of reals $x:\mathcal{A} \rightarrow (0,1)$ such that for all $A \in \A$, 
$$\prob{A} \leq x(A) \prod_{B \in \Gamma(A)}(1-x(B)).$$
Then the probability of avoiding all events in $\A$ is non-zero. More precisely,
$$\prob{\bigcap_{A \in \A} \overline{A}} \geq \prod_{A \in \A} (1 - x(A)) > 0.$$
\end{theorem}

A simple corollary of the LLL, called symmetric LLL, often suffices in several applications. 
In this version there is a uniform upper bound $p$ on the probability of each event and a uniform upper bound $d$ on the number of neighbors of each event in the dependency graph. This quantity $|\Gamma(A)|$ is also called the dependency degree of the event $A$. 
\begin{corollary}\cite{ErdosLovasz} \label{thm:nonconstructive-lll-symmetric}
If each event $A \in \A$ occurs with probability at most $p$ and has dependency
degree $|\Gamma(A)| \leq d$ such that $d \leq {1}/{ep} - 1$, then the probability that none of the events occur is positive.
\end{corollary}
\begin{proof}
Setting $x(A) = 1/({d+1})$ for all events $A \in \A$ shows that the conditions of Theorem~\ref{thm:nonconstructive-lll-general} are 
satisfied:
$$\prob{A}\le p \leq \frac{1}{e(d+1)} \leq \frac{1}{d+1} \left(1 - \frac{1}{d+1}\right)^{d}.$$
\end{proof}

The power of the symmetric version is well-demonstrated by showing a satisfiability result for $k$-CNF formulas, i.e., Boolean formulas in conjunctive normal form with $k$ variables per clause. This classic application of the LLL will help in understanding our and previous results and techniques and therefore will be a running example in the rest of the paper.

\begin{corollary} \label{cor:nonconstructive-ksat}
Every $k$-CNF formula in which every clause shares variables with at most $2^{k}/e - 1$ other clauses is satisfiable.
\end{corollary}
\begin{proof}
To apply the symmetric LLL (i.e., Corollary~\ref{thm:nonconstructive-lll-symmetric}) we choose the
probability space to be the product space of each variable being chosen
true or false independently with probability $1/2$. For each clause $C$ we
define an event $A_C$ that is said to occur if and only if clause
$C$ is not satisfied. Clearly, two events $A_C$ and $A_{C'}$ are
independent unless the clauses $C$ and $C'$ share variables. Now take
$G$ to be the graph on the events with edges between events $A_C$ and
$A_{C'}$ iff $C$ and $C'$ share variables. It is clear that each event
$A_C$ is mutually independent of its non-neighbors in $G$. By
assumption each event has at most $d \leq (2^{k}/e)
- 1$ neighbors. Moreover, the probability $p$ that a clause is not
satisfied by a random assignment is exactly $2^{-k}$. The requirement
$e p (d+1) \leq 1$ of Corollary~\ref{thm:nonconstructive-lll-symmetric} is
therefore met and hence we obtain that the probability that none of the events occur is
positive. The satisfiability of the $k$-CNF formula follows.
\end{proof}

\subsection{Previous work}
Algorithms for the LLL are often targeted towards one of two model
problems: $k$-CNF formula satisfiability and $k$-uniform hypergraph
2-coloring. Interesting in their own right, these problems capture the
essence of the LLL without many technicalities. Moreover, algorithms
for these problems usually lead to algorithms for more general
applications of the LLL \cite{CzumajS,MolloyReedBook,czumaj2000cnu}. As shown in
Section~\ref{sec:nonconstructive-LLL}, for the $k$-CNF formula
satisifiability problem, the LLL implies that every $k$-CNF formula in
which each clause shares variables with at most $2^k/e -1$ other
clauses has a satisfying assignment. 
Similarly, it can be shown that the vertices of a $k$-uniform hypergraph, in which each edge shares variables with at most $2^k/e-1$ other edges, can be colored using two colors so that no edge is monochromatic. 
The algorithmic objective
is to efficiently find such a 2-coloring (or a satisfying assignment in the case of $k$-CNF).

This question was first addressed by Beck in his seminal paper \cite{Beck}, where he
gave an algorithm for the hypergraph 2-coloring problem with  
dependency degree $O(2^{k/48})$. More precisely, he gave a
polynomial-time deterministic algorithm to find a 2-coloring of the
vertices of every $k$-uniform hypergraph in which each edge shares
vertices with $O(2^{k/48})$ other edges such that no edge is
monochromatic. Molloy and Reed~\cite{MolloyReedBook} showed that the dependency degree of this algorithm can be improved
to $2^{k/16}/4$. In the same volume in which Beck's paper appeared, Alon~\cite{Alon} gave a randomized parallel
version of Beck's algorithm that outputs a valid 2-coloring when the dependency degree
is at most $2^{k/500}$ and showed that this algorithm can be derandomized\footnote{Alon did not attempt to optimize the exponent but Srinivasan~\cite{Srinivasan} states that optimizing the bound would still lead to an exponent with several tens in the denominator.}. 
Since then, tremendous progress has been made on randomized LLL algorithms.
Nonetheless, prior to our work, Beck's and Alon's algorithms remained the best
deterministic and parallel algorithms for the (symmetric) LLL. 
 
For randomized algorithms and algorithms that require $k$ to be a fixed constant,
a long line of work improved the maximum achievable dependency degree and the generality
of the results culminating in the work of Moser and Tardos~\cite{MoserTardos}. Moser and Tardos \cite{MoserTardos} provided
a simple randomized (parallel) algorithm for the general LLL.
These results are summarized in Table~\ref{table:relatedwork}, and we discuss them next.

\medskip

\begin{table}
\begin{tabular}{ | c || l | c | c | c | } \hline
    & Max. Dep. Deg. $d$ & Det. & Par. & Remark \\ \hline \hline
Beck~\cite{Beck}  			& \, \ \ \ \ \ \ \ \ $O(2^{k/48})$ 	&	X	&&\\ \hline
Beck~\cite{MolloyReedBook} 			&	\ \,\ \ \ \ \ \ \ $O(2^{k/16})$ 	&	X	&& prev. best det. algorithm\\ \hline
Alon~\cite{Alon}   			& \ \ \ \,\ \ \ \ \ $O(2^{k/500})$	&	X	& X & prev. best det. par. algorithm\\ \hline
\color{Gray} \asabove &\color{Gray} \ \ \,\ \ \ \ \ \ $O(2^{k/8})$		&\color{Gray} X	&\color{Gray} X &\color{Gray} only for constant k,d\\ \hline
Srinivasan~\cite{Srinivasan} &	\ \ \ \ \,\ \ \ \ $O(2^{k/4})$		&   && \\ \hline
\asabove   			& \ \ \ \ \ \,\ \ \ $O(2^{k/10.3})$	&	  & X &\\ \hline
Moser~\cite{Moser08A} 	  & \ \ \ \ \ \ \,\ \ $O(2^{k/2})$		&&&    \\ \hline
\color{Gray}\asabove   			&\color{Gray} \ \ \ \ \ \ \ \,\ $O(2^{k/2})$	  &\color{Gray} X	&&\color{Gray} only for constant k,d\\ \hline
Moser~\cite{Moser08B} 	  & \,\ \ \ \ \ \ \ \ $O(2^{k})$      &&&		\\ \hline
\color{Gray}\asabove   			&\color{Gray} \ \,\ \ \ \ \ \ \ $O(2^{k})$      &\color{Gray} X &&\color{Gray} only for constant k,d\\ \hline
Moser, Tardos~\cite{MoserTardos}&	\ \ \ \ \ \ \ \;\;\;\, $(2^{k}/e - 1)$	&&&	\\ \hline
\asabove   			& $(1-\eps) \cdot (2^k/e - 1)$ &   & X &\\ \hline
\color{Gray}\asabove   			&\color{Gray} $(1-\eps) \cdot (2^k/e  - 1)$ &\color{Gray} X &\color{Gray} X & \color{Gray} only for constant k,d\\ \hline
Our work		        &	\ \ \ \ \ \ \ \;\;\;\, $(2^{k/(1+\eps)}/e - 1)$     & X & X &\\ \hline
\end{tabular}
\caption{Maximum dependency degrees achieved for $k$-CNF formulas by previous randomized, deterministic and parallel algorithms. Algorithmic results in gray have to assume a constant $k$ and a constant dependency degree $d$ in order to be efficient.}\label{table:relatedwork}
\end{table}

Alon~\cite{Alon} gave an algorithm that is efficient for a dependency 
degree of $O(2^{k/8})$ if one assumes that $k$ and therefore also the dependency degree
is bounded above by a fixed constant. Molloy and Reed \cite{MolloyReed} generalized Alon's
method to give efficient algorithms for a certain set-system model for
applications of the symmetric form of the LLL. Czumaj and Scheideler
\cite{CzumajS,czumaj2000cnu} consider the algorithmic problem for the
asymmetric version of the LLL. The asymmetric version of the LLL
addresses the possibility of 2-coloring the vertices of non-uniform
hypergraphs with no monochromatic edges. The next improvement in increasing
the dependency degree threshold was due to Srinivasan \cite{Srinivasan}. He gave
a randomized algorithm for hypergraph 2-coloring when the dependency
degree is at most $2^{k/4}$. Moser \cite{Moser08A} improved the
dependency degree threshold to $O(2^{k/2})$ using a variant of Srinivasan's
algorithm. Later, Moser \cite{Moser08B} achieved a significant
breakthrough improving the dependency degree threshold to $2^{k-5}$
using a much simpler randomized algorithm. Moser and Tardos~\cite{MoserTardos}
closed the small constant-factor gap to the optimal dependency degree $2^k/e$
guaranteed by the general LLL.

More importantly, Moser--Tardos \cite{MoserTardos} gave an algorithmic
framework for the general version of the LLL (discussed in Section
\ref{sec:algframework}) that minimally restricts the abstract 
LLL setting to make it amenable for algorithmic considerations. In this
framework they gave an efficient randomized algorithm for computing the structures implied by
the LLL. The importance of the framework stems from the fact that it
captures most of the LLL applications, thus directly providing algorithms for
these applications. Moser~\cite{Moser08A,Moser08B} and
Moser--Tardos~\cite{MoserTardos} also gave a derandomization of their algorithms
obtaining an algorithm that runs in $m^{O((1/\eps)d \log d)}$ time,
where $d$ is the maximum dependency degree and $m$ is the number of events. For the simpler $k$-CNF problem,
the running time of the deterministic algorithms can be improved to $m^{O(k^2)}$.
Nonetheless, this running time is polynomial only under the strong condition
that $k$ and the dependency degree are bounded by a fixed constant. 

The main open question that remained open was to obtain deterministic
algorithms that go beyond the initial results of Beck \cite{Beck} and
that are efficient for unbounded dependency degrees. We address this
question by giving new derandomizations of the Moser--Tardos algorithm. We give a
derandomization that works efficiently for the general version of the
LLL in the aforementioned algorithmic framework of Moser--Tardos
\cite{MoserTardos} assuming a mild $\eps$-slack in the LLL
conditions. As a corollary, we obtain an algorithm that runs in time
$\tilde{O}(m^{2(1+(1/\eps))})$ to find a satisfying assignment for a
$k$-CNF formula with $m$ clauses such that no clause shares variables
with more than $2^{k/(1+\eps)}/e$ other clauses, for any $\eps>0$.
We note that our $\eps$-slack assumption is in the exponent as opposed to the multiplicative slackness in the Moser--Tardos results (see Table \ref{table:relatedwork}). 
We also extend the randomized parallel algorithm of Moser--Tardos to
obtain an efficient deterministic parallel algorithm under the same
assumption thereby improving over Alon's algorithm with a dependency degree
of $O(2^{k/500})$.

\medskip \noindent {\bfseries Organization.} In Section
\ref{sec:preliminaries}, we describe the algorithmic
framework of Moser--Tardos for the LLL and their algorithm. In Section
\ref{sec:ourresults}, we state our results and their implications for
the $k$-CNF problem. In Section \ref{sec:ourtechniques}, we give an informal 
description of the new ideas in the paper. In
Section \ref{sec:partial_witness}, we formally define the major
ingredient in our derandomization: the partial witness structure. In
Section \ref{sec:deterministic_algorithm}, we give our sequential
deterministic algorithm and analyze its running time. Finally, in
Section \ref{sec:par}, we present our parallel algorithm and its
running time analysis.

\section{Preliminaries} \label{sec:preliminaries}

\subsection{Algorithmic Framework}\label{sec:algframework}

To get an algorithmic handle on the LLL, we move away from the abstract probabilistic setting of the original LLL. We impose some restrictions on the representation and form of the probability space under consideration. In this paper we follow the algorithmic framework for the LLL due to Moser--Tardos \cite{MoserTardos}. We describe the framework in this section.

The probability space is given by a finite collection of mutually independent discrete random variables $\P=\{P_1,\ldots,P_n\}$. Let $D_i$ be the domain of $P_i$, which is assumed to be finite.
Every event in a finite collection of events $\A=\{A_1,\ldots,A_m\}$ is determined by a subset  of $\P$. We define the {\bf variable set} of an event $A \in \A$ as the unique minimal subset $S \subseteq \P$ that determines $A$ and denote it by $\vbl(A)$. 

The {\bf dependency graph} $G = G_{\A}$ of the collection of events $\A$ is a graph on vertex set $\A$. The graph $G_{\A}$ has an edge between events $A, B \in \A$, $A\neq B$ if $\vbl(A) \cap \vbl(B) \neq \emptyset$.  For $A \in \A$ we denote the neighborhood of $A$ in $G$ by $\Gamma(A) = \Gamma_{\A}(A)$ and define $\Gamma^+(A)=\Gamma(A)\cup \{A\}$. Note that events that do not share variables are independent. 

It is useful to think of $\A$ as a family of ``bad'' events. The objective is to find a point in the probability space, or equivalently, an evaluation of the random variables from their respective domains, for which none of the bad events happens. We call such an evaluation a {\bf good evaluation}.

Moser and Tardos \cite{MoserTardos} gave a constructive proof of the
general version of the LLL in this framework
(Theorem~\ref{thm:general-lll}) using Algorithm 1 presented in
the next section.  This framework captures most known applications
of the LLL. 

\subsection{The Moser--Tardos Algorithm}\label{sec:MT}

Moser-Tardos \cite{MoserTardos} presented the very simple Algorithm 1 to find a good evaluation. 


\begin{figure}[!h]
\medskip \noindent {\bfseries Algorithm 1:\\ Sequential Moser--Tardos Algorithm}
\begin{enumerate}
\item For every $P\in \mathcal{P}$, $v_{P}\leftarrow$ a random evaluation of ${P}$.
\item While $\exists A\in \mathcal{A}$ such that $A$ happens on the
current evaluation $(P=v_{P}:\forall P\in \mathcal{P})$, do
\begin{enumerate}
\item Pick one such $A$ that happens (any arbitrary choice would work).
\item Resample $A$: For all $P\in \vbl(A)$, do
\begin{itemize}
\item $v_P\leftarrow$ a new random evaluation of $P$.
\end{itemize}
\end{enumerate}
\item Return $(v_P)_{P\in\mathcal{P}}$.
\end{enumerate}
\end{figure}

Observe that if the algorithm terminates, then it outputs a good evaluation. The following theorem from \cite{MoserTardos} shows that the algorithm is efficient if the LLL-conditions are met.

\begin{theorem}\cite{MoserTardos} \label{thm:general-lll}
Let $\A$ be a collection of events as defined in the algorithmic framework defined in Section \ref{sec:algframework}. If there exists an assignment of reals $x:\mathcal{A} \rightarrow (0,1)$ such that for all $A \in \A$,
$$\prob{A} \leq x'(A) := x(A) \prod_{B \in \Gamma(A)}(1-x(B)),$$
then the expected number of resamplings done by Algorithm 1 is at most $\sum_{A \in \A} \left(x(A)/(1-x(A))\right)$.
\end{theorem}

\section{Results} \label{sec:ourresults}

This section formally states the new results established in this paper. 

If an assignment of reals as stated in Theorem~\ref{thm:general-lll} exists, then we use such an assignment to define the following parameters\footnote{Throughout this paper $\log$ denotes the logarithm to base $2$.}:
\begin{itemize}
	\item $x'(A) :=  x(A) \prod_{B \in \Gamma(A)}(1-x(B))$.
	\item $D := \max_{P_i\in \mathcal{P}}\{|D_i|\}$.
	\item $\displaystyle M :=\max\left\{n,4m,2\sum_{A \in {\A}} \dfrac{2 |\vbl(A)|}{x'(A)}\cdot\dfrac{x(A)}{1-x(A)},\: \max_{A \in {\A}} \dfrac{1}{x'(A)}\right\}$.
	\item $w_{min}:=\min_{A\in\mathcal{A}}\{-\log{x'(A)}\}$.
	\item $\gamma = \frac{\log M}{\eps}$.
\end{itemize}

For the rest of this paper, we will use these parameters to express the running time of our algorithms. 

Our sequential deterministic algorithm assumes that for every event $A \in \A$, the conditional probability of occurrence of $A$ under any partial assignment to the variables in $\vbl(A)$, can be computed efficiently. This is the same complexity assumption as used in Moser--Tardos \cite{MoserTardos}. It can be 
further weakened to use pessimistic estimators. 

\begin{theorem} \label{thm:general_deterministic}
Let the time needed to compute the conditional probability $\prob{A
\:|\: \forall i \in I \: :\: P_i=v_i}$ for any $A \in \A$ and any
partial evaluation $(v_i \in D_i)_{i \in I}$, $I \subseteq [n]$, be at
most $t_C$. Suppose there is an $\eps \in (0,1)$ and an assignment of
reals $x:\mathcal{A}\rightarrow (0,1)$ such that for all $A \in \A$,
$$
\prob{A}\leq x'(A)^{1+\eps}=\left(x(A) \prod_{B \in
\Gamma(A)}(1-x(B))\right)^{1+\eps}.
$$

Then there is a deterministic algorithm that finds a good evaluation in time 
$$O\left(t_C \cdot \frac{DM^{3+2/\eps}\log^2{M}}{\eps^2 w_{min}^2}\right),$$
where the parameters $D$, $M$ and $w_{min}$ are as defined above.
\end{theorem}

We make a few remarks to give a perspective for the magnitudes of the parameters involved in our running time bound. As a guideline to reading the
results, it is convenient to think of $M$ as $\tilde{O}(m+n)$ and of $w_{min}$ as $\Omega(1)$. 

Indeed, $w_{min} = \Omega(1)$ holds whenever the $x(A)$'s are bounded away from one by a constant. For this setting we also have, without loss of generality\footnote{With $x(A)$ being
bounded away from one and given the $\eps$-slack assumed in our
theorems one can always reduce $\eps$ slightly to obtain a small
constant factor gap between $x'(A)$ and $\prob{A}$ (similar to
\cite{MoserTardos}) and then increase any extremely small $x(A)$ to at
least ${c}/{m}$ for some small constant $c > 0$. Increasing $x(A)$ in this way
only weakens the LLL-condition for the event $A$ itself.
Furthermore, the effect on the LLL-condition for 
any event due to the changed $(1-x(B))$ factors of one of its (at most
$m$) neighboring events $B$ accumulates to at most $(1-c/m)^m$
which can be made larger than the produced gap between $x'(A)$ and $\prob{A}$.}, 
that $x(A) =  \Omega(m^{-1})$. Lastly, the factor $(\prod_{B \in \Gamma(A)} (1 - x(B)))^{-1}$ is usually small;
e.g., in all applications using the symmetric LLL or the simple asymmetric version\cite{MolloyReed,MolloyReedBook} this factor is a constant. This makes $M$ at most a polynomial in $m$ and $n$. For most applications of the LLL this also makes $M$ polynomial in the size of the input/output. For all these settings our algorithms are efficient: The running time bound of our sequential algorithm is polynomial in $M$ and that of our parallel algorithm is polylogarithmic in $M$ using at most $M^{O(1)}$ many processors. 

Notable exceptions in which $M$ is not polynomial in the input size are the
problems in \cite{expLLL2}. For these problems $M$ is still $\tilde{O}(m+n)$
but the number of events $m$ is exponential in the number of variables $n$ and
the input/output size. For these settings, the problem of checking whether a given evaluation is good
is coNP-complete and obtaining a derandomized algorithm is an open question.

It is illuminating to look at the special case of $k$-CNF both in the
statements of our theorems as well as in the proofs, as many of the
technicalities disappear while retaining the essential ideas.  For
this reason, we state our results also for $k$-CNF. The magnitudes of the above parameters in the $k$-CNF applications are given by $x'(A)>1/de$, $D=2$, $M=\tilde{O}(m)$, and $w_{min}\approx k$.

\begin{corollary} \label{thm:kCNF_deterministic}
For any $\eps\in (0,1)$ there is a deterministic algorithm that finds a satisfying assignment 
to any $k$-CNF formula with $m$ clauses in which each clause shares variables 
with at most $2^{k/(1+\eps)}/e-1$ other clauses in time $\tilde{O}(m^{3+2/\eps})$.
\end{corollary}  

We also give a parallel deterministic algorithm. This algorithm
makes a different complexity assumption about the events, namely, that their
decision tree complexity is small. This assumption is quite general
and includes almost all applications of the LLL (except again for the 
problems mentioned in \cite{expLLL2}). They are an interesting alternative
to the assumption that conditional probabilities can be computed efficiently as used in the sequential algorithm. 

\begin{theorem}\label{thm:parallel_alg}
For a given evaluation, let the time taken by $M^{O(1)}$ processors to check the truth of an event $A\in \A$ be at most $t_{eval}$. Let $t_{MIS}$ be the time to compute the maximal independent set in an $m$-vertex graph using $M^{O(1)}$ parallel processors on an EREW PRAM. Suppose, there is an $\eps \in (0,1)$ and an assignment of reals $x:\mathcal{A}\rightarrow (0,1)$ such that for all $A \in \A$,
$$\prob{A}\leq x'(A)^{1+\eps}=\left(x(A) \prod_{B \in \Gamma(A)}(1-x(B))\right)^{1+\eps}.$$
If there exists a constant $c$ such that every event $A \in \A$ has {\bf decision tree complexity}\footnote{Informally,
we say that a function $f(x_1,\ldots,x_n)$ has decision tree complexity at most $k$ if we can determine its value by adaptively querying at most $k$ of the $n$ input variables.} at most $c\min \{-\log x'(A),\log{M}\}$, then there is a parallel algorithm that finds a good evaluation in time 
\[
O\left(\frac{\log{M}}{\eps w_{min}} (t_{MIS} +t_{eval}) + \gamma \log{D} \right)
\]
using $M^{O(({c}/{\eps})\log D)}$ processors.
\end{theorem}

The fastest known algorithm for computing the maximal independent set in an $m$-vertex graph using $M^{O(1)}$ parallel processors on an EREW PRAM runs in time $t_{MIS}=O(\log^2 m)$ \cite{alon-babai-itai, Luby}. Using this in the theorem, we get the following corollary for $k$-CNF.


\begin{corollary} \label{thm:kCNF_deterministic_parallel} 
For any $\eps\in (0,1)$ there is a deterministic parallel algorithm that uses $m^{O(1/\eps)}$ processors on an EREW PRAM
and finds a satisfying assignment to any $k$-CNF formula with $m$ clauses in which each clause shares variables 
with at most $2^{k/(1+\eps)}/e$ other clauses in time $O(\log^3{m}/\eps)$.
\end{corollary}

\section{Techniques}\label{sec:ourtechniques}
In this section, we informally describe the main ideas of our approach in the special context of $k$-CNF formulas and indicate how they generalize. Reading this section is not essential but provides intuition behind the techniques used for developing deterministic algorithms for the general LLL. For the sake of exposition in this section, we omit numerical constants in some mathematical expressions. Familiarity with the Moser--Tardos paper~\cite{MoserTardos} is useful but not necessary for this section.  

\subsection{The Moser--Tardos Derandomization}\label{sec:mtderand}

Let $F$ be a $k$-CNF formula with $m$ clauses.  We note immediately
that if $k > 1 + \log{m}$, then the probability that a random
assignment does not satisfy a clause is $2^{-k} \leq 1/(2m)$.  Thus
the probability that on a random assignment, $F$ has an unsatisfied
clause is at most $1/2$, and hence a satisfying assignment can be
found in polynomial time using the method of conditional probabilities
(see, e.g., \cite{MolloyReedBook}). Henceforth, we assume that $k \leq
1 + \log{m}$.  We also assume that each clause in $F$ shares variables
with at most $2^k/e-1$ other clauses; thus the LLL guarantees the
existence of a satisfying assignment. 

To explain our techniques we first need to outline the deterministic
algorithms of Moser and of Moser--Tardos which work in polynomial
time, albeit only for $k=O(1)$. Consider a table $T$ of values: for
each variable in $\P$ the table has a sequence of values, each picked
at random according to its distribution. We can run Algorithm 1 using
such a table: instead of randomly sampling afresh each time a new
evaluation for a variable is needed, we pick its next unused value
from $T$. The fact that the randomized algorithm terminates quickly in
expectation (Theorem~\ref{thm:general-lll}), implies that there exist
small tables (i.e., small lists for each variable) on which the
algorithm terminates with a satisfying assignment. The deterministic
algorithm finds one such table.

The constraints to be satisfied by such a table can be described in terms
of \emph{witness trees}: for a run of the randomized algorithm,
whenever an event is resampled, a witness tree ``records'' the
sequence of resamplings that led to the current resampling. We will
not define witness trees formally here; see \cite{MoserTardos} or
Section~\ref{sec:partial_witness} for a formal definition.  We say
that a witness (we will often just use ``witness'' instead of
``witness tree'') is \emph{consistent} with a table, if this witness
arises when the table is used to run Algorithm 1.  If the algorithm
using a table $T$ does not terminate after a small number of
resamplings, then it has a large consistent witness certifying this
fact. Thus if we use a table which has no large consistent witness,
the algorithm should terminate quickly.

The deterministic algorithms of Moser and of Moser--Tardos compute a
list $L$ of witness trees satisfying the following properties.
\begin{enumerate}  
\item Consider an arbitrary but fixed table $T$. If no witness in $L$ is consistent 
with $T$, then there is no large witness tree consistent with $T$. 
\item The expected number of witnesses in $L$ consistent with a random
table is less than $1$.  This property is needed in order to apply the
method of conditional probabilities to find a small table with which
no tree in $L$ is consistent.
\item The list $L$ is of polynomial size. This property is necessary
for the method of conditional probabilities to be efficient.
\end{enumerate}

We now describe how these properties arise naturally while using Algorithm
1 and how to find the list $L$. In the context of $k$-CNF formulas
with $m$ clauses satisfying the degree bound, Moser (and also
Moser--Tardos when their general algorithm is interpreted for $k$-CNF)
prove two lemmas that they use for derandomization.  The
\emph{expectation lemma} states that the expected number of large
(size at least $\log m$) consistent witness trees (among all possible
witness trees) is less than $1/2$ (here randomness is over the choice
of the table).  At this point we could try to use the method of
conditional probabilities to find a table such that there are no large
witness trees consistent with it.  However there are infinitely many
witness trees, and so it is not clear how to proceed by this method.

This difficulty is resolved by the \emph{range lemma} which states
that if for some $u$, no witness tree with size in the range $[u, ku]$
is consistent with a table, then no witness tree of size at least $u$
is consistent with the table. Thus, the list $L$ is the set of witness
trees of size in the range $[u, ku]$. Now one can find the required
table by using the method of conditional probabilities to exclude all
tables with a consistent witness in $L$.  The number of witnesses in
$L$ is $m^{\Omega(k^2)}$.  To proceed by the method of conditional
probabilities we need to explicitly maintain $L$ and find values
for the entries in the table so that none of the witnesses in $L$
remains consistent with it.

Thus, the algorithm of Moser (and respectively Moser--Tardos) works in
polynomial time only for constant $k$. Clearly, it is the size of $L$
that is the bottleneck towards achieving polynomial running time for
$k=\omega(1)$. One possible way to deal with the large size of $L$
would be to maintain $L$ in an implicit manner, thereby using a small
amount of space. We do not know how to achieve this. We solve this
problem in a different way, by working with a new (though closely
related) notion of witness trees, which we explain next.

\subsection{Partial Witness Trees}

For a run of the Moser--Tardos randomized algorithm using a table $T$,
for each resampling of an event, we get one witness tree consistent
with $T$.  Given a consistent witness tree of size $ku+1$, removing
the root gives rise to up to $k$ new consistent witnesses, whose union
is the original witness minus the root.  Clearly one of these new
subtrees has size at least $u$.  This proves their range lemma. The
range lemma is optimal for the witness trees. That is, for a given $u$
it is not possible to reduce the multiplicative factor of $k$ between
the two endpoints of the range $[u, ku]$.

We overcome this limitation by introducing \emph{partial witness
trees}, which have properties similar to those of witness trees, but
have the additional advantage of allowing a tighter range lemma.  The
only difference between witness trees and partial witness trees is
that the root, instead of being labeled by a clause $C$ (as is the
case for witness trees), is labeled by a \emph{subset} of variables
from $C$.  Now, instead of removing the root to construct new witness
trees as in the proof of the Moser--Tardos range lemma, each subset of
the set labeling the root gives a new consistent partial witness tree.
This flexibility allows us to prove the range lemma for the smaller
range $[u, 2u]$.  The number of partial witness trees is larger than
the number of witness trees because there are $2^km$ choices for the
label of the root (as opposed to $m$ choices in the case of witness
trees) since the root may be labeled by any subset of variables in a
clause.  But $2^k \leq 2m$, as explained at the beginning of
Section 4.1.  Thus for each witness tree there are at most $2^k \leq
2m$ partial witnesses and the expectation lemma holds with similar
parameters for partial witnesses as well.  The method of conditional
probabilities now needs to handle partial witness trees of size in
the range $[\log{m}, 2 \log{m}]$, which is the new $L$.  The number of
partial witnesses in this range is $m^{\Omega(k)}$, which is still too
large. The next ingredient brings this number down to a manageable
size.

\subsection{$\eps$-slack}
By introducing an $\eps$-slack, that is, by making the slightly stronger
assumption that each clause intersects at most $2^{(1-\eps)k}/e$ other
clauses, we can prove a stronger expectation lemma: The expected
number of partial witnesses of size more than $(4\log{m})/\eps k$ is
less than $1/2$. Indeed, the number of labeled trees of size $u$ and
degree at most $d$ is less than $(ed)^u\leq 2^{(1-\eps)ku}$ (see
\cite{knuth1969countingtrees}). Thus the number of partial witnesses
of size $u$ is less than $2^k m 2^{(1-\eps)ku}$, where the factor $2^k
m$ ($\leq 2m^2$) accounts for the number of possible labels for the
root. Moreover, the probability that a given partial witness tree of
size $u$ is consistent with a random table is $2^{-k(u-1)}$ (as opposed
to $2^{-ku}$ in the case of a witness tree). This is proved in a
similar manner as for witness trees. Thus the expected number of
partial witnesses of size at least $\gamma=4\log{m}/\eps k$ consistent
with a random table is at most
\[
\sum_{u\geq \gamma} 2^k m 2^{(1-\eps)ku} \cdot 2^{-k(u-1)} \leq \sum_{u\geq \gamma} 2^{2k} m 2^{-\eps ku} \leq \sum_{u\geq \gamma} 4m^3 2^{-\eps ku} \leq 1/2.
\]

Now, by the new expectation and range lemmas it is sufficient to consider
partial witnesses of size in the range $[(4\log{m})/\eps k, (8\log{m})/\eps
k]$.  The number of partial witnesses of size in this range is polynomial in
$m$; thus the list $L$ of trees that the method of conditional
probabilities needs to maintain is polynomial in size.

\subsection{General Version}
More effort is needed to obtain a deterministic algorithm for the general version of the LLL. Here, the events are allowed to have significantly varying probabilities of occurrence and unrestricted structure.

One issue is that an event could possibly depend on all $n$
variables. In that case, taking all variable subsets of a label for
the root of a partial witness would give up to $2^n$ different
possible labels for the roots. However, for the range lemma to hold
true, we do not need to consider all possible variable subsets for the
root; instead, for each root event $A$ it is sufficient to have a
pre-selected choice of $2\vbl(A)$ labels. This pre-selected choice of
labels $\mathbb{B}_A$ is fixed for each event $A$ in the beginning.

The major difficulty in derandomizing the general LLL is in finding a list $L$ satisfying the three properties mentioned earlier for applying the method of conditional probabilities. The range lemma can still be applied. However, the existence of low probability events with (potentially) many neighbors may lead to as many as $O(m^u)$ partial witnesses of size in the range $[u,2u]$. Indeed, it can be shown that there are instances in which there is no setting of $u$ such that the list $L$ containing all witnesses of size in the range $[u,2u]$ satisfies properties (2) and (3) mentioned in Section 4.1.

The most important ingredient for working around this in the general
setting is the notion of \emph{weight of a witness tree}. The weight
of a tree is the sum of the weights of individual vertices; more
weight is given to those vertices whose corresponding bad events have
smaller probability of occurrence. Our deterministic algorithm for the
general version finds a list $L$ that consists of partial witnesses
with weight (as opposed to size) in the range $[\gamma, 2\gamma]$,
where $\gamma$ is a number depending on the problem. It is easy to
prove a similar range lemma for weight based partial witnesses which
guarantees property (1) for this list. Further, the value of $\gamma$
can be chosen so that the expectation lemma of Moser and Tardos can be
adjusted to lead to property (2) for $L$. Unfortunately one cannot
prove property (3) by counting the number of partial witnesses using
combinatorial enumeration methods as in \cite{Moser08B}. This is due
to the possibility of up to $O(m)$ neighbors for each event $A$ in the
dependency graph. However, the strong coupling between weight and
probability of occurrence of bad events can be used to obtain property
(3) directly from the expectation lemma.

\subsection{Parallel Algorithm} \label{subsec:parallel}
For the parallel algorithm, we use the technique of
limited-independence spaces, or more specifically $k$-wise
$\delta$-dependent probability spaces due to Naor and
Naor\cite{NaorNaor} and its extensions \cite{Evenetal, Charietal}.
This is a well-known technique for derandomization.  The basic idea
here is that instead of using perfectly random bits in the randomized
algorithm, we use random bits chosen from a limited-independence
probability space.  For many algorithms it turns out that their
performance does not degrade when using bits from such a probability
space; but now the advantage is that these probability spaces are much
smaller in size and so one can enumerate all the sample points in them
and choose a good one, thereby obtaining a deterministic algorithm.
This tool was applied by Alon~\cite{Alon} to give a deterministic
parallel algorithm for $k$-uniform hypergraph 2-coloring and other
applications of the LLL, but with much worse parameters than ours. Our
application of this tool is quite different from the way Alon uses it:
Alon starts with a random 2-coloring of the hypergraph chosen from a
small size limited independence space; he then shows that at least one
of the sample points in this space has the property that the
monochromatic hyperedges and almost monochromatic hyperedges form
small connected components.  For such a coloring, one can alter it
locally over vertices in each component to get a valid 2-coloring.

In contrast, our algorithm is very simple (we describe it for $k$-CNF;
the arguments are very similar for hypergraph 2-coloring and for the
general LLL): recall that for a random table, the expected number of
consistent partial witnesses with size in the range $[(4\log{m})/\eps
k, (8\log{m})/\eps k]$ is at most $1/2$ (for the case of $k$-CNF).
Each of these partial witnesses uses at most $((8 \log{m})/{\eps k})
\cdot k=((8 \log{m})/{\eps})$ entries from the table.  Now, instead of using a completely
random table, we use a table chosen according to a $\left(8
\log{m}/\eps\right)$-wise independent distribution (i.e., any subset
of at most $(8 \log{m})/\eps$ entries has the same joint distribution
as in a random table).  So any partial witness tree is consistent with
the new random table with the same probability as before.  And hence
the expected number of partial witnesses consistent with the new
random table is still at most $1/2$.  But now the key point to note is
that the number of tables in the new limited independence distribution
is much smaller and we can try each of them in parallel until we
succeed with one of the tables.  To make the probability space even
smaller we use $k$-wise $\delta$-dependent distributions, but the idea
remains the same.  Finally, to determine whether a table has no
consistent partial witness whose size is at least $(4\log{m})/\eps k$,
we run the parallel algorithm of Moser--Tardos on the table.

In order to apply the above strategy to the general version, we require that the number of variables on which witnesses depend be small, and hence the number of variables on which    
events depend should also be small.  In our general parallel algorithm we relax this to some extent: instead of requiring that each event 
depend on few variables, we only require that the decision tree complexity of the event be small.  The idea behind the proof remains the same.

\section{The Partial Witness Structure}\label{sec:partial_witness}

In this section we define the partial witness structure and their weight.
We then prove the new range lemma using these weights. 

\subsection{Definitions}

For every $A \in \A$ we fix an arbitrary rooted {\bf binary variable
splitting} $\mathbb{B}_A$. It is a binary tree in which all vertices
have labels which are nonempty subsets of $\vbl(A)$: The root of
$\mathbb{B}_A$ is labeled by $\vbl(A)$ itself, the leaves are labeled
by distinct singleton subsets of $\vbl(A)$ and every non-leaf vertex
in $\mathbb{B}_A$ is labeled by the disjoint union of the labels of
its two children. This means that every non-root non-leaf vertex is
labeled by a set $\{v_{i_1},\ldots,v_{i_k}\}$, $k\geq 2$ while its
children are labeled by $\{v_{i_1},\ldots,v_{i_{j}}\}$ and
$\{v_{i_{j+1}},\ldots,v_{i_k}\}$ for some $1\leq j\leq k-1$. Note that
$\mathbb{B}_A$ consists of $2|\vbl(A)|-1$ vertices. We abuse the
notation $\mathbb{B}_A$ to also denote the set of labels of the
vertices of this binary variable splitting. The binary variable
splitting is not to be confused with the (partial) witness tree which
we define next. The elements from $\mathbb{B}_A$ will solely be used
to define the possible labels for the roots of partial witness trees.
An example of a binary variable splitting $\mathbb{B}_A$ can be found in Figure \ref{fig:binarysplitting}.

\begin{figure}[!ht]
\centering
\includegraphics[scale=0.7]{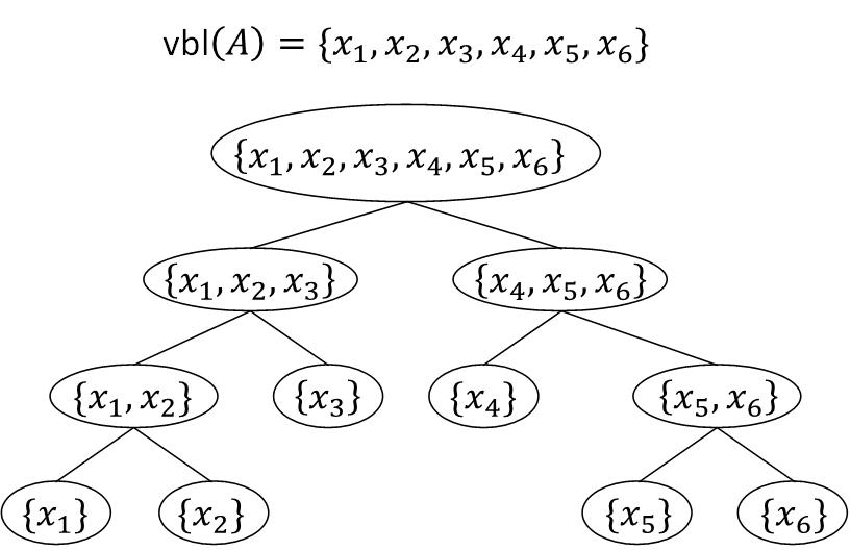}
\caption{A binary variable splitting for an event $A$ that depends on variables $x_1,x_2,x_3,x_4,x_5,x_6$.}
\label{fig:binarysplitting}
\end{figure}

A {\bf partial witness tree $\tau_S$} is a finite rooted tree whose
vertices apart from the root are labeled by events from $\A$ while the
root is labeled by some subset $S$ of variables with $S \in
\mathbb{B}_R$ for some $R \in \A$. Each child of the root must be
labeled by an event $A$ that depends on at least one variable in $S$
(thus a neighbor of the corresponding root event); the children of
every other vertex that is labeled by an event $B$ must be labeled
either by
$B$ or by a neighboring event of $B$, i.e., its label should be from
$\Gamma^+(B)$. Define
$V(\tau_S)$ to be the set of vertices of $\tau_S$. For notational
convenience, we use
$\overline{V}(\tau_S):=V(\tau_S)\setminus\{\mathrm{Root}(\tau_S)\}$
and denote the label of a vertex $v \in \overline{V}(\tau_S)$ by
$[v]$.

A {\bf full witness tree} is a special case of a partial witness where
the root is the complete set $\vbl(A)$ for some $A\in \mathcal{A}$. In
such a case, we relabel the root with $A$ instead of $\vbl(A)$. Note
that this definition of a full witness tree is the same as the one of
the witness trees in \cite{MoserTardos}.

Define the weight of an event $A\in \A$ to be
$w(A)=-\log{x'(A)}$. Define the {\bf weight of a partial witness tree}
$\tau_{S}$ as the sum of the weights of the labels of the vertices in
$\overline{V}(\tau_S)$, i.e.,
\[
w({\tau_S}):= \sum_{v\in \overline{V}(\tau_S)}w([v])=-\log{\left(\prod_{v\in \overline{V}(\tau_S)} x'([v]) \right)}.
\]

The {\bf depth} of a vertex $v$ in a witness tree is the distance of $v$ from the root in the witness tree. We say that a partial witness tree is {\bf proper} if for every vertex $v$, all children of $v$ have distinct labels. 

Similar to \cite{Moser08A}, we will control the randomness used by the
algorithm using a {\bf table} of evaluations, denoted by $T$. It is
convenient to
think of $T$ as a matrix.  This table contains one row for each
variable in $\P$. Each row contains evaluations for its variable. Note
that the number of columns in the table could possibly be infinite. In
order to use such a table in the algorithm, we maintain a pointer
$t_i$ for each variable $P_i\in \P$ indicating the column containing
its current value used in the evaluation of the events. We denote the
value of $P_i$ at $t_i$ by $T(i,t_i)$. If we want to resample an
evaluation for $P_i$, we increment the pointer $t_i$ by one, and use
the value at the new location.

We call a table $T$ a {\bf random table} if, for all variables $P_i\in \P$ and all positions $j$, the entry $T(i,j)$ is picked independently at random according to the distribution of $P_i$. It is clear that running Algorithm 1 is equivalent to using a random table to run Algorithm 2 below. 


\begin{figure}[!ht] \label{alg:2}

\medskip\noindent{\bf Algorithm 2:\\ Moser-Tardos Algorithm with input table\\} 

\noindent Input: Table $T$ with values for variables \\
Output: An assignment of values for variables so that none of the events in $\mathcal{A}$ happens 
\begin{enumerate}
\item For every variable $P_i \in \P$: Initialize the pointer $t_{i} = 1$.
\item While $\exists A\in \mathcal{A}$ that happens on the current
assignment (i.e., $\forall P_i \in \P: P_i = T(i,t_{i})$) do
\begin{enumerate}
\item Pick one such $A$.
\item Resample $A$: For all $P_i \in \vbl(A)$ increment $t_{i}$ by one.
\end{enumerate}
\item Return $ \forall P_i \in \P : P_i = T(i,t_{i})$.
\end{enumerate}
\end{figure}

In the above algorithm, Step 2(a) is performed by a fixed arbitrary deterministic procedure. This makes the algorithm well-defined. 

Let $C:\mathbb{N}\rightarrow \mathcal{A}$ be an ordering of the events (with repetitions), which we call the {\bf event-log}. Let the ordering of the events as they have been selected for resampling in the execution of Algorithm 2 using a table $T$ be denoted by an event-log $C_T$. Observe that $C_T$ is partial if the algorithm terminates after a finite number of resamplings $t$; i.e., $C_T(i)$ is defined only for $i\in \{1,2,\ldots,t\}$.

Given an event-log $C$, associate with each resampling step $t$ and
each $S\in \mathbb{B}_{C(t)}$, a partial witness tree $\tau_{C}(t,S)$
as follows. Define $\tau_{C}^{(t)}(t,S)$ to be an isolated root vertex
labeled $S$. Going backwards through the event-log, for each
$i=t-1,t-2,\ldots,1$: (i) if there is a non-root vertex $v\in
\tau_{C}^{(i+1)}(t,S)$ such that $C(i)\in \Gamma^+([v])$, then among
all such vertices choose the one whose distance from the root is
maximum (break ties arbitrarily) and attach a new child vertex $u$ to
$v$ with label $C(i)$, thereby obtaining the tree
$\tau_{C}^{(i)}(t,S)$, (ii) else if $S\cap\vbl(C(i))$ is non-empty,
then attach a new child vertex to the root with label $C(i)$ to obtain
$\tau_{C}^{(i)}(t,S)$, (iii) else, set
$\tau_{C}^{(i)}(t,S)=\tau_{C}^{(i+1)}(t,S)$. Finally set
$\tau_C(t,S)=\tau_C^{(1)}(t,S)$.

Note that if $S= \vbl(A) \in \mathbb{B}_A$ then $\tau_C(t,S)$ is a
full witness tree with root $A$. For such a full witness tree, our
construction is the same as the construction of \emph{witness trees
associated with the log} in \cite{MoserTardos}.

We say that the partial witness tree $\tau_S$ {\bf occurs} in
event-log $C$ if there exists $t\in \mathbb{N}$ such that for some
$A\in \A$ such that $S\in \mathbb{B}_A$, $C(t)=A$ and
$\tau_S=\tau_{C}(t,S)$. An illustrating example of these definitions can be found in Figure \ref{fig:witnesstree}.

\begin{figure}
\centering
\begin{tabular}{cc}
\includegraphics[width=0.40\textwidth]{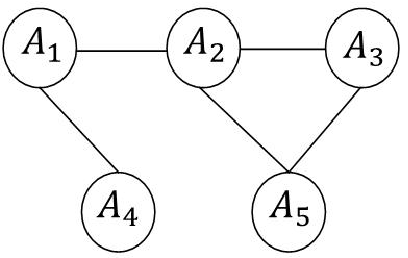} & 
\includegraphics[width=0.30\textwidth]{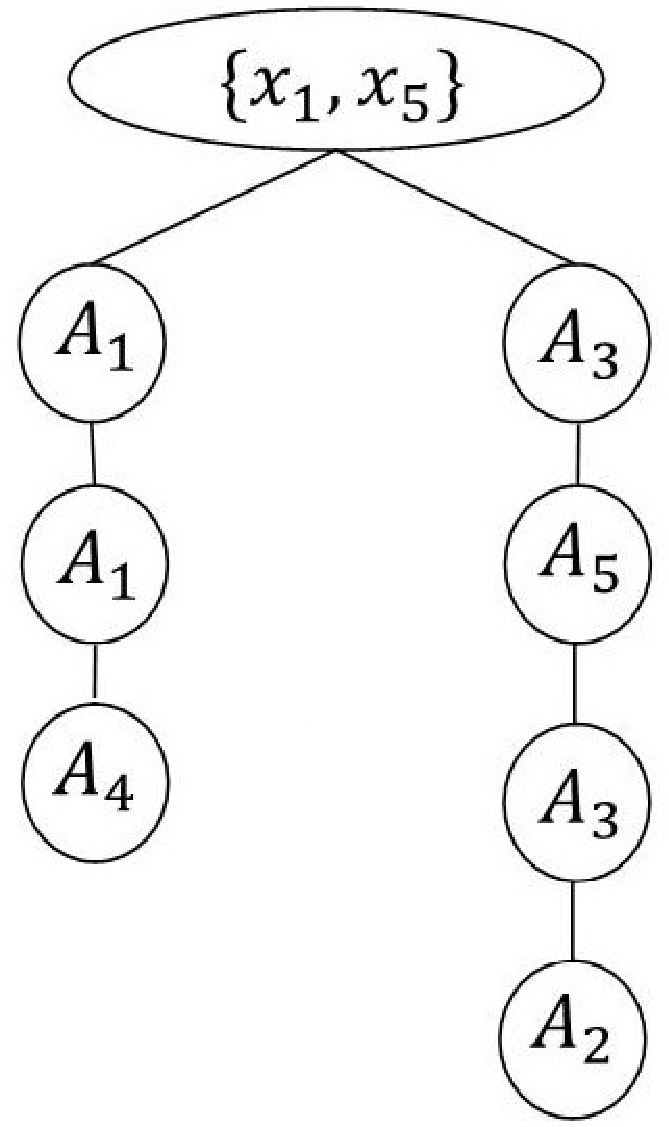}\\
Dependency graph & Partial witness tree $\tau_C(9,\{x_1,x_5\})$
\end{tabular}
\caption{The dependency graph and an example of a partial witness tree constructed from the event-log $C = A_2, A_3, A_5, A_4, A_1, A_3, A_1, A_5, A_2, \ldots$  where $\vbl(A_1)=\{x_1,x_2,x_3\}$, $\vbl(A_2)=\{x_1,x_4,x_5\}$, $\vbl(A_3)=\{x_4,x_5,x_6\}$, $\vbl(A_4)=\{x_3,x_7\}$, $\vbl(A_5)=\{x_4,x_6\}$. Note that the last occurrence of the event $A_5$ is not added to the witness since it does not share a variable with the variable subset $\{x_1,x_5\} \subset \vbl(A_2)$ that was selected as a root.}
\label{fig:witnesstree}
\end{figure}

For a table $T$, a {\bf $T$-check} on a partial witness tree $\tau_S$
uses table $T$ as follows: In decreasing order of depth, visit the
non-root vertices of $\tau_S$ and for a vertex with label $A$, take
the first unused value from $T$ for each $x\in \vbl(A)$ and check if
the resulting evaluation makes $A$ happen. The {\bf $T$-check passes}
if all events corresponding to vertices apart from the root, happen
when checked. We say that a partial witness tree is {\bf consistent
with a table $T$} if the $T$-check passes on the partial witness tree.

Most of the above definitions are simple extensions of the ones given in \cite{MoserTardos}. 

\subsection{Properties}
In this section we state and prove two important properties of the
partial witness tree which will be useful in obtaining the
deterministic sequential and parallel algorithms.

The following lemma proves that given a witness tree, one can use the
$T$-check procedure to exactly determine which values were used in the
resamplings that lead to this witness tree.

\begin{lemma}\label{lemma:witness-check}
For a fixed table $T$, if a partial witness tree $\tau_S$ occurs in the event-log $C_T$, then
\begin{enumerate}
\item $\tau_S$ is proper.
\item $\tau_S$ is consistent with $T$.
\end{enumerate}
\end{lemma}
\begin{proof}
The proof of this lemma is essentially Lemma 2.1 in \cite{MoserTardos} and 
included here for completeness. 

Since $\tau_S$ occurs in $C_T$, there exists some time instant $t$
such that for $S\in \mathbb{B}_{C_T(t)}$,
$\tau_S=\tau_{C_T}(t,S)$. For each $v\in \overline{V}(\tau_S)$, let
$d(v)$ denote the depth of vertex $v$ and let
$q(v)$ denote the largest value $q$ with $v$
contained in $\tau_{C_T}^{(q)}(t)$. We observe that $q(v)$ is 
the time instant in which $v$
was attached to $\tau_{C_T}(t,S)$
by the procedure constructing $\tau_{C_T}(t,S)$.

If $q(u)<q(v)$ for vertices $u,v \in \overline{V}(\tau_S)$ and
$\vbl([u])$ and $\vbl([v])$ are not disjoint, then
$d(u)>d(v)$. Indeed, when adding the vertex $u$ to
$\tau_{C_T}^{(q(u)+1)}(t)$ we attach it to $v$ or to another vertex of
equal or greater depth. Therefore, for any two vertices $u,v\in
\overline{V}(\tau_S)$ at the same depth $d(u)=d(v)$, $[u]$ and $[v]$
do not depend on any common variables, that is the labels in every
level of $\tau_S$ form an independent set in $G$. In particular
$\tau_S$ must be proper.

Now consider a non-root vertex $v$ in the partial witness tree
$\tau_S$. Let $P_i\in \vbl([v])$. Let $\mathcal{D}(i)$ be the set of
vertices $w\in \tau_S$ with depth greater than that of $v$ such that
$[w]$ depends on variable $P_i$.

When the $T$-check considers the vertex $v$ and uses the next unused
evaluation of the variable $P_i$, it uses the evaluation
$T(i,|\mathcal{D}(i)|)$. This is because the witness check visits the
vertices in order of decreasing depth and among the vertices with
depth equal to that of $v$, only $[v]$ depends on $P_i$ (as we proved
earlier that vertices with equal depth are variable disjoint). So the
$T$-check must have used values for $P_i$ exactly when it was
considering the vertices in $\mathcal{D}(i)$.

At the time instant of resampling $[v]$, say $t_v$, Algorithm 2
chooses $[v]$ to be resampled which implies that $[v]$ happens before
this resampling. For $P_i\in \vbl([v])$, the value of the variable
$P_i$ at $t_v$ is $T(i,{|\mathcal{D}(i)|})$. This is because the
pointer for $P_i$ was increased for events $[w]$ that were resampled
before the current instance, where $w\in \mathcal{D}(i)$. Note that
every event which was resampled before $t_v$ and that depends on $[v]$
would be present at depth greater than that of $v$ in $\tau_S$ by
construction. Hence, $\mathcal{D}(i)$ is the complete set of events
which led to resampling of $P_i$ before the instant $t_v$.

As the $T$-check uses the same values for the variables in $\vbl([v])$
when considering $v$ as the values that led to resampling of $[v]$, it
must also find that $[v]$ happens.
\end{proof}

Next, we prove a range lemma for partial witnesses, improving the range to a factor of two. 
 
\begin{lemma}\label{lemma:range}
If a partial witness tree of weight at least $\gamma$ occurs in the
event-log $C_T$ and every vertex $v$ in the tree has weight at most
$\gamma$, then a partial witness tree of weight $\in [\gamma,2\gamma)$
occurs in the event-log $C_T$.
\end{lemma}
\begin{proof}
The proof is by contradiction. Consider a least weight partial
witness tree whose weight is at least $\gamma$ that occurs in the
event-log $C_T$, namely $\tau_S=\tau_{C_T}(t,S)$ for some $t$, $S\in
\mathbb{B}_{A}$ where $A=C_T(t)$. A witness tree with weight at least $\gamma$
exists by assumption and because there are only finitely many choices for $t$ and $S$
there exists also a least weight such tree. 
Suppose, for the sake of
contradiction, that $w({\tau_S})\geq 2\gamma$. We may assume that
Root$(\tau_S)$ has at least one child, otherwise, the weight of the
tree is zero. We have two cases:

Case (i): Root$(\tau_S)$ has only one child $v$. Let $t'$ be the
largest time instant before $t$ at which $[v]$ was resampled. Note
that this resampling of $[v]$ corresponds to the child $v$ of the root
of $\tau_S$. Now, consider the partial witness tree
$\tau_S'=\tau_{C_T}(t',S'=\vbl([v]))$. Since $\tau_S'$ contains one
less vertex than $\tau_S$, $w({\tau_S'}) < w({\tau_S})$. Also, since
the weight of any vertex $v$ in the tree is at most $\gamma$ we get
that $w({\tau_S'}) = w({\tau_S}) - w([v]) \geq \gamma$. Finally, by
definition of $\tau_S'$, it is clear that $\tau_S'$ occurs in the
event-log $C_T$. Thus, $\tau_S'$ is a counterexample of smaller weight
contradicting our choice of $\tau_S$.

Case (ii): Root$(\tau_S)$ has at least two children. Since the
 labeling clauses of these children have pairwise disjoint sets of
 variables and they have to share a variable with $S$, we have that
 $S$ consists of at least $2$ variables. Thus, it also has at least two
 children in the variable splitting $\mathbb{B}_A$. In $\mathbb{B}_A$,
 starting from $S$, we now explore the descendants of $S$ in the
 following way, looking for the first vertex whose children $S_L$ and
 $S_R$ reduce the weight of the tree, i.e., $0 < w({\tau_{S_L}}),
 w(\tau_{S_R})<w(\tau_S)$, where $\tau_{S_L}=\tau_{C_T}(t,S_L)$ and
 $\tau_{S_R}=\tau_{C_T}(t,S_R)$: if a vertex $S_L$ reduces the weight
 of the tree without making it zero (i.e.,
 $0<w(\tau_{S_L})<w(\tau_S)$), then its variable disjoint sibling
 $S_R$ must also reduce the weight of the tree; on the other hand, if
 a vertex $S_L$ reduces the weight of the tree to zero, then its
 sibling $S_R$ cannot reduce the weight of the tree. Suppose $S_L$
 reduces the weight to zero, then we explore $S_R$ to check if its
 children reduce the weight. It is easy to see that this stops at the
 latest when $S_L$ and $S_R$ are leaves in $\mathbb{B}_A$.

By definition, both $\tau_{S_L}$ and $\tau_{S_R}$ occur in the
event-log $C_T$. Since we pick the first siblings $S_L$ and $S_R$ (in
the breadth first search) which reduce the weight, their parent $S'$
is such that $w(\tau_{S'})\geq w(\tau_S)$, where
$\tau_{S'}=\tau_{C_T}(t,S')$. We are considering only those $S'$ such
that $S'\subseteq S$. This implies that $w(\tau_{S'})\leq
w(\tau_{S})$. Hence, $w(\tau_{S'})=w(\tau_{S})$ and for every vertex
that has label $A$ in $\tau_{S}$, one can find a unique vertex labeled
by $A$ in $\tau_{S'}$ and vice-versa. Further, $S'$ is the disjoint
union of $S_L$ and $S_R$; therefore, for each vertex with label $A$ in
$\tau_{S'}$, one can find a unique vertex labeled by $A$ either in
$\tau_{S_L}$ or $\tau_{S_R}$.

As a consequence, we have that for every vertex with label $A$ in
$\tau_{S}$, one can find a unique vertex labeled by $A$ either in
$\tau_{S_L}$ or $\tau_{S_R}$. Hence, $w(\tau_{S_L})+w(\tau_{S_R})\geq
w(\tau_S)$ and therefore, $\max\{w(\tau_{S_L}),w(\tau_{S_R})\}\geq
w(\tau_{S})/2\geq \gamma$. So, the witness with larger weight among
$\tau_{S_L}$ and $\tau_{S_R}$ has
weight at least $\gamma$ but less than
that of $\tau_S$. This contradicts our choice of $\tau_S$.
\end{proof}

\section{Deterministic Algorithm}\label{sec:deterministic_algorithm}
In this section we describe our sequential deterministic algorithm and prove Theorem \ref{thm:general_deterministic}. 

For the rest of the paper we define a set of {\bf forbidden witnesses} $F$ which contains all partial witness trees with weight between $\gamma$ and $2\gamma$. 
We define a table to be a {\bf good table} if no forbidden witness is consistent with it. With these definitions we can state our deterministic algorithm.

\begin{figure}[!ht]
\medskip \noindent {\bfseries Algorithm 3:\\ Sequential Deterministic Algorithm}
\begin{enumerate}
\item Enumerate all forbidden witnesses in $F$.
\item Construct a good table $T$ via the method of conditional probabilities:\\
For each variable $p\in \P$, and for each $j$, $0\leq j\leq
2\gamma/w_{min}$, do
\begin{itemize}
\item Select a value for $T(p,j)$ that minimizes the expected number
of forbidden witnesses that are consistent with $T$ when all
entries in the table chosen so far are fixed and the yet to be chosen
values are random.
\end{itemize}
\item Run Algorithm 2 using table $T$ as input.
\end{enumerate}
\end{figure}

We next give a short overview of the running time analysis of Algorithm 3 before embarking on the proof of Theorem \ref{thm:general_deterministic}.

The running time of Algorithm 3 depends on the time to construct a
good table $T$ by the method of conditional probabilities. To construct such a
table efficiently, we prove that the number of forbidden witnesses is
small (polynomial in $M$) using Lemma~\ref{lemma:enumerate}. Further,
we need to show that the method of conditional probabilities indeed
constructs a good table. We show this by proving in Lemma~\ref{lemma:expected-no-of-witnesses}
that the expected number of forbidden witnesses that are consistent 
with $T$ initially (when all values are random) is smaller than one. 
This invariant is maintained by the method of conditional probabilities
resulting in a fixed table with less than one (and therefore no)
forbidden witnesses consistent with it. By Lemmas
\ref{lemma:witness-check} and \ref{lemma:range}, it follows that no witness
of weight more than $\gamma$ occurs when Algorithm 2 is run on the table
$T$. Finally, the maximum number of vertices in a partial witness tree of
weight at most $\gamma$ is small. This suffices to show that the size of 
table $T$ is small and thus Algorithm 3 is efficient.

\begin{lemma}\label{lemma:expected-no-of-witnesses}
The expected number of forbidden witnesses consistent with a 
random table $T$ is less than $1/2$.
\end{lemma}
\begin{proof}
For each event $A\in {\A}$, let $\Upsilon_{A}$ and $\Upsilon_{A}'$ be
the set of partial and respectively full witness trees in $F$ with
root from $\mathbb{B}_A$. With this notation the expectation in
question is exactly:
\[
\sum_{A\in {\A}}\sum_{\tau \in \Upsilon_{A}} \prob{\tau \text{ is consistent with } T}.
\]
Note that according to Lemma \ref{lemma:witness-check}, a partial
witness tree is consistent with a table $T$ if and only if it passes the
$T$-check. Clearly, the probability that a witness $\tau$ passes the
$T$-check for the random table $T$ is $\prod_{v\in
\overline{V}(\tau)}\prob{[v]}$ (recall that $\overline{V}(\tau)$
denotes the set of 
non-root vertices in $\tau$). Using this and the assumption in
Theorem \ref{thm:general_deterministic} that $\prob{[v]} \leq
x'([v])^{1+\eps}$ we get that the expectation is at most 
\[
E := \sum_{A\in {\A}}\sum_{\tau \in \Upsilon_{A}} \prod_{v\in \overline{V}(\tau)}x'([v])^{1+\eps}.
\]
To relate this to the full witness trees considered in
\cite{MoserTardos}, we associate with every partial witness tree
$\tau$ (in $\Upsilon_A$) a full witness tree $\tau'$ (in
$\Upsilon'_A$) by replacing the root subset $S\in \mathbb{B}_A$ with
the full set $\vbl(A)$. Note that the weights of $\tau$ and $\tau'$ are
the same (as is the quantity $\prod_{v \in \overline{V}(\tau)}
x'([v])^{1+\eps}$). Note also that every full witness
tree has at most $|\mathbb{B}_A|$ partial witness trees associated
with it. Hence, we can rewrite the expression to get
\begin{align*}
E&\leq \sum_{A\in {\A}} |\mathbb{B}_A| \sum_{\tau \in \Upsilon_{A}'} \prod_{v\in \overline{V}(\tau)}x'([v])^{1+\eps}\\
&\leq \sum_{A\in {\A}} |\mathbb{B}_A| \sum_{\tau \in \Upsilon_{A}'} \left(\prod_{v\in \overline{V}(\tau)}x'([v])\right)2^{-\gamma\eps},\\
\end{align*}
where the last expression follows because, for $\tau\in \Upsilon_A'$, we have
\begin{align*}
w({\tau}) = -\log{\prod_{v\in \overline{V}(\tau)} x'([v]) } &\geq \gamma,\\
\implies \prod_{v\in \overline{V}(\tau)}x'([v])\leq 2^{-\gamma}.\\
\end{align*}

Next we transition from partial to full witness trees by including the
root again (and going from $\overline{V}$ to $V$).

$$E \leq \sum_{A\in {\A}} \frac{|\mathbb{B}_A|}{x'(A)}\left(\sum_{\tau \in \Upsilon_{A}'} \prod_{v\in V(\tau)}x'([v])\right)2^{-\gamma\eps}.$$

Now we can use the following result of Moser--Tardos (Section 3 in \cite{MoserTardos}) that bounds the expected number of full witnesses with root $A$:

$$\sum_{\tau \in \Upsilon_{A}'} \prod_{v\in V(\tau)}x'([v]) \leq \frac{x(A)}{1-x(A)}.$$

Their proof makes use of a Galton--Watson process that randomly
generates proper witness trees with root $A$ (note that by Lemma
\ref{lemma:witness-check} all partial witness trees are proper). Using
this,
\begin{align*}
E &\leq \sum_{A\in \overline{\A}} \frac{|\mathbb{B}_A|}{x'(A)}\cdot\left(\frac{x(A)}{1-x(A)}\right)2^{-\gamma\eps}\\
&< \frac{M}{2} 2^{-\gamma\eps} \leq \frac{1}{2}.
\end{align*}

Here the penultimate inequality follows from the fact that
$|\mathbb{B}_A| < 2|\vbl(A)|$ and the definition of $M$, and the last
inequality follows from the choice of $\gamma=(\log{M})/\eps$.
\end{proof}

Owing to the definition of forbidden witnesses \emph{via weights},
there is an easy way to count the number of forbidden witnesses using
the fact that their expected number is small.


\begin{lemma}\label{lemma:enumerate}
The number of witnesses with weight at most $2\gamma$ is at most $O(M^{2(1+1/\eps)})$. In particular, the number of forbidden witnesses is less than $M^{2(1+1/\eps)}$.
\end{lemma}
\begin{proof}
Each forbidden witness $\tau\in F$ has weight $w(\tau)\leq 2\gamma$ and thus
\begin{align*}
|F| (2^{-2\gamma})^{(1+\eps)} &\leq \sum_{\tau \in F} (2^{-w(\tau)})^{(1+\eps)}\\
&= \sum_{\tau \in F} \left(\prod_{v\in \overline{V}(\tau)} x'([v])\right)^{(1+\eps)}\\
&= E \leq \frac{M}{2} 2^{-\gamma\eps} \leq \frac{1}{2}.
\end{align*}
Here, the final line of inequalities comes from the proof of Lemma \ref{lemma:expected-no-of-witnesses}. Therefore the number of forbidden witnesses is at most 
$$|F| \leq \left(\frac{M}{2} 2^{-\gamma\eps}\right) \cdot 2^{2\gamma(1+\eps)} =\left(\frac{M}{2}\right)2^{\gamma(2+\eps)}\leq \frac{1}{2} M^{2(1+1/\eps)}.$$

Using the same argument with any $\gamma'$ instead of $\gamma$ shows that the number of witnesses with weight in $[\gamma',2\gamma']$ is at most $({M}/{2}) \cdot 2^{\gamma'(2+\eps)}$. Since this is exponential in $\gamma'$ the total number of witnesses with weight at most $2\gamma$ is dominated by a geometric sum which is $O(M^{2(1+1/\eps)})$.
\end{proof}

We are now ready to prove Theorem \ref{thm:general_deterministic}.

\begin{proof}[Proof of Theorem \ref{thm:general_deterministic}]


We first describe how the set of forbidden witnesses in the first step of the deterministic algorithm (Algorithm 3) is obtained. 

\paragraph{Enumeration of witnesses.} We enumerate all witnesses of weight at most $2\gamma$ and then discard the ones with weight less than $\gamma$. According to Lemma \ref{lemma:enumerate}, there are at most $M^{2(1+1/\eps)}$ witnesses of weight at most $2\gamma$ and each of them consists of at most $x_{\max} = (2\gamma/w_{min})+1 = (2\log{M})/(\eps w_{min})+1$ vertices. In our discussion so far we did not need to consider the order of children of a node in our witness trees. However, for the enumeration it will be useful to order the children of each node from left to right. We will build witnesses by attaching nodes level-by-level and from left to right. We fix an order on the events according to their weights, breaking ties arbitrarily, and use the convention that all witnesses are represented so that for any node its children from left to right have increasing weight. We then say a node $v$ is {\em eligible} to be attached to a witness $\tau$ if in the resulting witness $\tau'$ the node $v$ is the deepest rightmost leaf in $\tau'$. With this convention the enumeration proceeds as follows: 

As a preprocessing step for every event $A$ we sort all the events in $\Gamma^+(A)$ according to their weight in $O(m^2 \log m)$ time. Then, starting with $W_1$, the set of all possible roots, we incrementally compute all witnesses $W_x$ having $x = 1, \ldots, x_{\max}$ nodes and weight at most $2 \gamma$. To obtain $W_{x+1}$ from $W_{x}$ we take each witness $\tau \in W_x$ and each node $v \in \tau$ and check for all $A \in \Gamma^+([v])$ with weight more than the current children of $v$, in the order of increasing weight whether a node $v'$ with $[v']=A$ is eligible to be attached to $\tau$ at $v$. If it is eligible, and the resulting new witness $\tau'$ has weight at most $2\gamma$, then we add $\tau'$ to $W_{x+1}$. It is clear that in this way we enumerate all forbidden witnesses without producing any 
witness more than once. 

We now analyze the time required by the above enumeration procedure. 
We write down each witness explicitly, taking $O(x_{\max}\log{M})$ time and space per witness. For each witness it takes linear (in the number of nodes) 
time to find the nodes with eligible children. Note that attaching children to a node in the order of increasing weight guarantees that at most one attachment attempt per node fails due to large weight. Thus, the total time to list all forbidden witnesses is at most $O(x_{\max} M^{2(1+1/\eps)}\log{M})$.

\paragraph{Finding a good table.} The running time to find a good table $T$ using the method of conditional probabilities as described in Algorithm 3 can be bounded as follows: For each of the $n$ variables, the table $T$ has $2\gamma/w_{min} = x_{\max}$ entries to be filled in. For each of those entries at most $D$ possible values need to be tested. For each value we compute the
conditional expectation of the number of forbidden witnesses that are 
consistent with the partially filled in table $T$ by computing the conditional
probability of each forbidden witness $\tau \in F$ to pass the $T$-check
given the filled in values and summing up these probabilities.
This can be done by plugging in the fixed values into each of the at most $x_{\max}$ nodes of $\tau$
similar to the $T$-check procedure, computing the conditional probability in $t_{C}$ time and computing the product
of these conditional probabilities. Thus, the total time to compute $T$ is at most
$$O(n \cdot x_{max} \cdot D \cdot |F| \cdot  x_{\max} \cdot t_{C}) = O\left(\frac{DM^{3+2/\eps}\log^2 {M}}{\eps^2 w_{min}^2}t_C\right).$$

To complete the proof we show that the running time of the sequential algorithm on a table $T$ obtained by Step 2 of the deterministic algorithm is at most $O\left(m^2 \cdot x_{\max} \cdot t_C\right)$:

First, we note that by running the sequential algorithm using table
$T$, none of the forbidden witnesses can occur in the event-log
$C_T$. This is because the table is obtained by the method of
conditional probabilities: In the beginning of the construction of the
table, when no value is fixed, the expected number of forbidden
witnesses that occur in the event-log is less than $1/2$ as proved in
Lemma \ref{lemma:expected-no-of-witnesses}. This invariant is
maintained while picking values for variables in the table. Thus, once
all values are fixed, the number of witness trees in $F$ that occur in
the event-log $C_T$ is still less than $1/2$ and hence zero. 

This
implies that the sequential algorithm with $T$ as input resamples each
event $A\in \A$ at most $x_{\max}$ times. Indeed, if some event $A \in \A$ is
resampled more than $x_{\max}$ times, then $A$ occurs in the event-log $C_T$
at least $x_{\max}$ times. Now, the weight of the partial witness tree
associated with the last instance at which $A$ was resampled, would be
at least $x_{\max} w_{min}$ which is more than $2\gamma$. According to Lemma~\ref{lemma:range}, 
which is applicable since $\gamma=(\log{M})/\eps$
is larger than the maximum weight event, there would also be a
forbidden witness of weight between $\gamma$ and $2\gamma$ occurring
in $C_T$, a contradiction. Therefore, the number of resamplings done
by Algorithm 2 is $O\left(m \cdot x_{\max} \right)$ and the total running time 
for Algorithm 2 using table $T$ is $O\left(m^2 \cdot x_{\max} \cdot t_C\right)$: the additional factor $m  \cdot t_C$ comes from the time
needed to find an event that happens. This running time is
smaller than the upper bound for the time needed to find a good
table $T$.

This shows that Algorithm 3 terminates in the stated time bound. Lastly, the correctness of the
algorithm follows directly from the fact that the algorithm only terminates if a good assignment
is found. 
\end{proof}

From the general deterministic algorithm it is easy to obtain the corollary regarding $k$-CNF by using the standard reduction to the symmetric LLL and plugging in the optimal values for the parameters. 

\begin{proof} [Proof of Corollary \ref{thm:kCNF_deterministic}]
For a $k$-CNF formula with clauses $\A=\{A_1,\ldots,A_m\}$, for each clause $A\in \A$ we define an event $A$ and say that the event happens if the clause is unsatisfied. Further, each variable appearing in the formula picks values uniformly at random from $\{0,1\}$. Then, for every event $A$, $\prob{A}=2^{-k}$. As remarked in Section~\ref{sec:mtderand}, we may assume that $k < \log m$, for otherwise the problem becomes simple. If $d$ is the maximum number of clauses that a clause shares its variables with, setting $x(A)=1/d$ for all $A\in \A$, we obtain that $x'(A)>{1}/{de}$. The condition that $d\leq 2^{k/(1+\eps)}/e$ then implies for all events $A$ that $\prob{A}\leq x'(A)^{1+\eps}$ as required by the LLL-condition. Therefore, we use parameters $t_C=O(k)$, $w_{min}\approx k$, $D=2$, $|\vbl(A)|=k$ and obtain $M=O(n+m+mk+d)=O(m\log{m})$. With these parameters the corollary follows directly from Theorem~\ref{thm:general_deterministic}. 
\end{proof}

\section{Parallel Algorithm}\label{sec:par}
In this section we present an efficient parallel algorithm (outlined in Sec.~\ref{subsec:parallel}) and
analyze its performance, thereby proving Theorem~\ref{thm:parallel_alg}. 

In the design of our sequential algorithm, we used Algorithm 2 as a subroutine which takes an input table, and uses 
it to search for an assignment for which none of the bad events happens. 
This reduced the problem to finding a good input table. For designing the parallel algorithm, Moser--Tardos already
provided the parallel counterpart of Algorithm 2, and so what remains is to find a good table. 
Our algorithm relies on the following observation: Instead of sampling the values in the table independently at random, if we choose it from a distribution that is a $(k,\delta)$-approximation of the original distribution (for appropriate $k$ and $\delta$), the algorithm behaves as if the values in the table had been chosen independently at random (Proposition \ref{prop:fooling_complexity}).  The support of a $(k,\delta)$-approximation can be chosen to be small and
can be generated fast in parallel, so this gives us a small set of tables which is guaranteed to contain at least one table on which the algorithm terminates quickly (Lemma~\ref{lemma:parallel-expected-witnesses}).  Our algorithm runs the Moser--Tardos parallel algorithm
on each of these tables in parallel, and stops as soon as one of the tables leads to a good evaluation.

We begin by describing the two ingredients that we will need. 

\subsection{Limited independence probability spaces} \label{subsec:limitedind}

We need the notion of $(k,\delta)$-approximate distributions to describe our algorithm.

\begin{definition}
$(k,\delta)$-approximations \cite{Evenetal}:  Let $\mathcal{S}$ be a 
product probability distribution on a finite domain $S_1 \times S_2 \times \ldots \times S_s$ given 
by mutually independent random variables $X_1, \ldots, X_s$, where $X_i \in S_i$. 
For positive integer $k$ and constant $\delta \in (0,1)$, a probability distribution $\mathcal{Y}$ on 
$S_1 \times S_2 \times \ldots \times S_s$ is said to be a $(k,\delta)$-approximation
of $\mathcal{S}$ if the following holds.  For every $I \subseteq [s]$ such that $|I| \leq k$, and every 
$v \in S_1 \times S_2 \times \ldots \times S_s$ we have
\begin{align*}
|\Pr_{\mathcal{S}}[v_I]-\Pr_{\mathcal{Y}}[v_I]| \leq \delta,
\end{align*}
where $\Pr_{\mathcal{S}}[v_I]$ denotes the probability that for a
random vector $(x_1, \ldots, x_s)$ chosen according to the probability
distribution $\mathcal{S}$, we get $x_i = v_i$ for $i \in I$; the definition of 
$\Pr_{\mathcal{Y}}[v_I]$ is analogous.
\end{definition}



The support $Y$ of a $(k,\delta)$-approximation $\mathcal{Y}$ of $\mathcal{S}$ can be 
constructed efficiently in parallel. We use the construction described in \cite{Evenetal} (which in turn uses \cite{NaorNaor}).
This construction builds a $(k,\delta)$-approximation to a product space with $t$ variables
with a support size of $|Y|=\mathrm{poly}(2^k, \log t, \delta^{-1})$. %
The construction can be
parallelized to run in time $O(\log{t} + \log{k} +  \log{1/\delta} + \log{D})$  
using $\text{poly}(2^k/\delta) t D$  processors, where $D$ is again the maximum domain size for a variable. 

For our algorithm we want approximately random tables of small size. 
More formally we will work with tables containing at most $\ceil{\gamma/w_{min}}$ columns. So, we set $t = n \cdot \ceil{\gamma/w_{min}}$ and $S_1 \times S_2 \times \ldots \times S_s = (D_1 \times D_2 \times \ldots \times D_n)^{\ceil{\gamma/w_{min}}}$. We furthermore set $k = 2c \gamma$, $\delta^{-1} = 3 M^{2+2/\eps}D^{2c\gamma}$ and 
$\mathcal{S}$ to be the distribution obtained by independently sampling each entry in the table according to its distribution.
For these values and recalling that $\gamma = (\log{M})/\epsilon$ the support $Y$ of the $(k, \delta)$-approximation $\mathcal{Y}$ obtained by the construction mentioned above has size $\mathrm{poly}(2^{2c \gamma}, \log{(n \cdot\gamma/w_{min})}, 3 M^{2+2/\eps}D^{2c\gamma}) = M^{O((c/\epsilon)\log D)}$, and it can be constructed in parallel in time 
$O(\gamma \log{D} +  \log(1/w_{min}))$ using 
$M^{O((c/\epsilon)\log D)}$ processors. 

\subsection{Decision trees}

In Theorem~\ref{thm:parallel_alg}, our assumption about how the events depend on the variables was in terms of 
decision tree complexity. In this section we recall the definition of decision trees, and show some simple properties
needed in the sequel. 

Let $S = D_1 \times \ldots \times D_n$, and let $f : S \rightarrow \{0, 1\}$ be a Boolean function. We denote the
elements of $S$ by  $(x_1, x_2, \ldots, x_n)$ where $x_i \in D_i$ for $1 \leq i \leq n$. A decision tree for 
computing $f(x_1, x_2, \ldots, x_n)$ is a rooted tree $T$, where each internal vertex of the tree is labeled by 
one of the variables from $\{x_1,\ldots,x_n\}$, and each leaf is labeled by $0$ or $1$. An internal vertex
labeled by $x_i$, has $|D_i|$ children, with their corresponding edges being labeled by distinct elements from 
$D_i$.
To compute $f(x_1, x_2, \ldots, x_n)$, the execution of $T$ determines a root-to-leaf path as follows: starting at 
the root we query the value of the variable labeling a vertex and follow the edge to the child which is labeled 
by the answer to the query. When we reach the leaf, we output the label of the leaf. The complexity of a decision tree is its depth. The decision tree complexity of a function $f$ is the depth of 
the shallowest decision tree computing $f$. 

\begin{prop}\label{prop:fooling_complexity}
Let $S = D_1 \times \ldots \times D_n$ be a product space of finite domains of size at most $D = \max_i |D_i|$, let $\mathcal{P}$ be an independent product distribution on $S$ and let $f,f_1,f_2:S \rightarrow \{0,1\}$ be Boolean functions on $S$. 
\begin{enumerate}
	\item If $f_1$ and $f_2$ have decision tree complexity $k_1$ and $k_2$ respectively, then the decision tree complexity of $f_1 \wedge f_2$ is at most $k_1+k_2$. 
	\item If $f$ has decision tree complexity at most $k$ then every $(k,\delta)$-approximation $\mathcal{Y}$ of $\mathcal{P}$ is  $D^k\delta$-indistinguishable from $\mathcal{P}$, i.e.,
$$|E_{\mathcal{Y}}(f) - E_{\mathcal{P}}(f)| \leq D^k\delta.$$
\end{enumerate}
\end{prop}
\begin{proof}
For the first claim we recall that a function $f$ having decision tree
complexity at most $k$ is equivalent to saying that we can determine
$f(x)$ for $x \in S$ by adaptively querying at most $k$ coordinates of
$x$. If this is true for $f_1$ and $f_2$ with decision tree complexity
$k_1$ and $k_2$ respectively then we can evaluate $f_1(x)
\wedge f_2(x)$ by adaptively querying at most $k_1 + k_2$ components
of $x$. Therefore the conjunction has decision tree complexity at
most $k_1 + k_2$.

For the second claim, we fix a decision tree for $f$ with depth at most $k$. Each one of the leaf-to-root paths in this tree corresponds to a partial assignment of values to at most $k$ components, and this assignment determines the value of $f$. The expectation of $f$ under any distribution is simply the sum of the probabilities of the paths resulting in a 1-evaluation at the leaf. Switching from a completely independent distribution to a $k$-wise independent distribution does not change these probabilities since the partial assignments involve at most $k$ variables. Similarly switching to a $(k,\delta)$-approximation changes each of these probabilities by at most $\delta$.  There are at most $D^k$ paths resulting in a 1-evaluation which implies that the  deviation of the expectation is at most $D^k\delta$.
\end{proof}

The following lemma shows that using a $(k,\delta)$-approximate distribution instead of the original one does not change the performance of Algorithm 2 if the events have low decision tree complexity: 
\begin{lemma}\label{lemma:parallel-expected-witnesses}
Suppose that there exists a constant $c$ such that every event $A \in \A$ has decision tree complexity at most 
$c\min \{-\log x'(A),\log{M}\}$. Let $k=2c\gamma$ and $\delta^{-1} = 3 M^{2+2/\eps}D^{2c\gamma}$.
The expected number of forbidden witnesses consistent with a table $T$
that was created by a $(k,\delta)$-approximation for the distribution
of random tables is at most $1/2+1/3<1$.
\end{lemma}
\begin{proof}
The event that a partial witness $\tau \in F$ is consistent with $T$ is exactly the conjunction of events $[v]$, $v \in \overline{V}(\tau)$.  Using Proposition~\ref{prop:fooling_complexity}, the decision tree complexity of this event is at most 
\begin{align*}
\sum_{v \in \overline{V}(\tau)} c \min\{\log M,-\log x'([v])\} &\leq c \sum_{v \in \overline{V}(\tau)} -\log x'([v]) \leq 2c\gamma,
\end{align*}
where the last inequality follows because by definition, forbidden
witnesses have weight at most $2 \gamma$.  Lemma~\ref{lemma:expected-no-of-witnesses} shows that using the original
independent distribution $\P$, the expected number of forbidden
witnesses occurring is at most $1/2$. The second claim of
Proposition~\ref{prop:fooling_complexity} proves that switching to a
$(k,\delta)$-approximation changes this expectation by at most
$D^k\delta = 1/(3 M^{2+2/\eps})$ for each of the $|F|$ witnesses.  To
complete the proof, observe that by Lemma~\ref{lemma:enumerate} we
have $|F| \leq M^{2+2/\eps}$.
\end{proof}

\subsection{The parallel algorithm and its analysis}

We can now describe our parallel algorithm. 

\begin{figure}[!ht]
\medskip\noindent{\bf Algorithm 4:\\ Parallel Deterministic Algorithm}
\begin{enumerate}
\item \label{step:parallel_loop0} Construct a small set of tables $Y$ which form the support of a $(k,\delta)$-approximate independent distribution $\mathcal{Y}$ using the construction mentioned in Sec.~\ref{subsec:limitedind}.
\item \label{step:parallel_loop1} For each table $T \in Y$ do in parallel:
	\begin{enumerate}
		\item \label{step:parallel_ph2_beg} For every variable
		$P_i \in \P$: initialize the pointer $t_i = 1$.

		\item While $\exists A\in \A$ that happens when
		$\forall P_i \in \P: P_i = T(i,t_i)$, do

			\begin{itemize}
				\item Compute, in parallel, a maximal
				independent set $I$ in the subgraph of
				$G_{\A}$ induced by the events that
				happen on the current assignment.
				\item Resample all $A \in I$ in parallel:
				For all $P_i \in \bigcup_{A \in I}
				\vbl(A)$, increment $t_i$ by one.
                                \item If $t_i = \ceil{\gamma/w_{min}} + 1$ (one more than the total number of samples for $P_i$ in a good table), then halt this thread of computation. 
		\end{itemize}
	\end{enumerate}
\item \label{step:parallel_ph2_end} Once a valid assignment is found
using one of the tables, output it and terminate.
\end{enumerate}
\end{figure}


\begin{proof} [Proof of Theorem~\ref{thm:parallel_alg}]
We use Algorithm 4 to obtain a good evaluation. We already saw in Section~\ref{subsec:limitedind} that the support $Y$ of the $(k, \delta)$-approximation to the random distribution of tables in Step 1 can be generated efficiently within the time and the number of processors claimed. We now show that these resources also suffice for the rest of the steps in the algorithm. 

Lemma~\ref{lemma:parallel-expected-witnesses} guarantees that there is a table $T \in Y$ for which there is no forbidden witness consistent with it. Steps \ref{step:parallel_ph2_beg}--\ref{step:parallel_ph2_end} are the same as the parallel algorithm in \cite{MoserTardos}. 
We will show that on table $T$ this algorithm terminates within at most $\ceil{\gamma/w_{min}}$ steps: By using Lemma~4.1 of \cite{MoserTardos}, if the algorithm runs for $i$ iterations, then there exists a consistent
witness of height $i$. Such a witness has weight at least $i w_{min}$. 
We know from Lemma~\ref{lemma:range} and Lemma~\ref{lemma:witness-check} that no witness of weight more than $\gamma$ can
occur since otherwise a forbidden witness would be consistent with $T$. Hence we have $i \leq
\gamma/w_{min}$. This means that the thread for table $T$ does not attempt to increment the pointer $t_i$ beyond 
$\gamma/w_i$ on table $T$, and so this thread terminates with a good evaluation.  
Each of these $i$ iterations takes time $t_{eval}$ to evaluate all $m$ events and time $t_{MIS}$ to compute
the independent set on the induced dependency subgraph of size at most
$m$. This proves that after creating the probability space
$\mathcal{Y}$, the algorithm terminates in $O((t_{MIS} + t_{eval})\gamma/w_{min})$ time and the termination criterion guarantees
correctness. Adding this to $O(\gamma \log{D} + \log(1/w_{min}))$, the time to construct $Y$, we get that the total time the algorithm takes
is $O((t_{MIS} + t_{eval})\gamma/w_{min} + \gamma \log{D})$. The number of processors needed for the loop is
bounded by $M^{O(1)}$ for each of the $|Y|$ parallel
computations and thus $M^{O((c/\epsilon)\log D)}$ in total. 
\end{proof}

Again it is easy to obtain the $k$-CNF result as a corollary of the general algorithm:

\begin{proof} [Proof of Corollary \ref{thm:kCNF_deterministic_parallel}]
We apply the LLL in the same way to $k$-CNF as in the proof of Corollary \ref{thm:kCNF_deterministic}. Again we assume without loss of generality that $k = O(\log n)$ and again get $M=O(m k)$ and $w_{min}\approx k$. Since each clause depends only on $k$ variables a decision tree complexity of $O(k)$ is obvious. Finally using Theorem \ref{thm:parallel_alg} and an algorithm of Alon, Babai and Itai \cite{alon-babai-itai} or Luby~\cite{Luby} to compute the maximal independent set in time $t_{MIS} = O(\log^2 m)$ leads to the claimed running time.
\end{proof}

\section{Conclusion}

Moser and Tardos \cite{MoserTardos} raised the open question for a
deterministic LLL algorithm. We address this question and give a deterministic
parallel algorithm that works under nearly the same conditions as its
randomized versions. 

All known deterministic or (randomized) parallel
algorithms need a slack in the LLL conditions (see Table~\ref{table:relatedwork}).
It remains open to remove those $\eps$-slacks. Obtaining deterministic 
constructions for the problems in \cite{expLLL2} is another interesting
open question.

\paragraph{Acknowledgements}
We are thankful to the anonymous reviewers who caught an error in
an earlier version and whose comments greatly improved the
presentation of this paper. We also thank Aravind Srinivasan for inspiration
and Salil Vadhan and David Karger for many helpful comments.

\addcontentsline{toc}{section}{References}
\bibliographystyle{abbrv}    
\bibliography{lll}
\end{document}